\documentclass[12pt,titlepage]{article}

\usepackage{amsmath,amssymb,amsthm}
\usepackage{booktabs}
\usepackage{xcolor}
\usepackage{natbib}
\usepackage{hyperref}
\usepackage[capitalize]{cleveref}
\usepackage{setspace}

\newtheorem{lemma}{Lemma}

\newtheorem{definition}{Definition}

\DeclareMathOperator*{\argmax}{arg\,max}

\begin{document}

\centerline{\Large{Revealed Rationality:}}
\centerline{\large{Label-Free Evaluation and Regularization from Representation Theorems}}

\bigskip

\centerline{By Isaiah Andrews\footnote{First version: February 21, 2026.  This version: August 5, 2026. Department of Economics, Massachusetts Institute of Technology, and NBER, iandrews@mit.edu.  I thank Jiafeng Chen, Lukasz Kowalik, and Sendhil Mullainathan for useful discussions, and Claude and GPT for excellent research assistance.}}

\bigskip

\centerline{Abstract}
\begin{spacing}{1.25}
{\noindent\small{
Representation theorems in decision theory establish that behavior satisfies certain
axioms if and only if it can be rationalized by a well-defined objective. I argue that this ``if and
only if'' structure provides a potentially useful foundation for label-free evaluation and regularization of LLMs
 and other AI systems. Axiom compliance can be checked from the model's own
responses to synthetic choice problems, with no external labels or human feedback, and
the penalties are readily computable. 
Because the axioms are necessary and sufficient, the resulting checks exhaust the implications of the relevant rationality standard for the elicited data: a model that passes cannot be rejected on rationality grounds by any further test of the same data.
I discuss three instantiations: probabilistic 
coherence via a theorem of de Finetti, preference rationality via Afriat's theorem, and 
subjective expected utility via a theorem of \citet{echenique2015savage}, each yielding a continuous
penalty that is zero whenever behavior can be rationalized. Since coherence
does not restrict which objective rationalizes behavior, these penalties complement rather than replace other evaluation and training signals.\\

}
}
\end{spacing}
\section{Introduction}\label{sec:intro}

Large language models are increasingly used in contexts that require them to
express beliefs, rank alternatives, and recommend actions.  Recent work
suggests both that the rationality of these outputs varies substantially
across models and that it is responsive to design choices.
\citet{tak2026sparks}, for instance, test whether large language models satisfy the von
Neumann-Morgenstern axioms (completeness, transitivity, continuity,
independence) and find widespread violations, but also find that models using
additional reasoning tokens show substantially improved axiom compliance.
\citet{chadwick2025dutch} observe that incoherent probability assessments 
and intransitive preference orderings in
large language models are directly exploitable by constructing Dutch books (combinations of
bets guaranteed to extract money from the model) and money pumps (sequences of trades
which extract value) respectively.\footnote{Exploitability arguments of this kind have a long history as normative foundations for rationality axioms; see \citet{gustafsson2022} for a recent treatment.}
Together, these findings illustrate both that rationality violations can be
consequential and that they are amenable to intervention.

A broader diagnostic literature reinforces these observations.
\citet{chen2023pnas} test whether GPT-3.5 choices from budget sets satisfy
axioms from revealed preference theory, finding high but imperfect
compliance.  
\citet{hagendorff2023llm} show that earlier generations of LLMs depart from 
rational choice, but that these violations are smaller for later model generations.  
\citet{raman2024steer} develop a benchmark organized around a taxonomy of rationality elements, and construct ``report cards'' summarizing performance across models.
\citet{mazeika2025utility} find that larger models largely
satisfy expected utility axioms, but sometimes exhibit unappealing emergent goals.
\citet{qiu2026bayesian} find that LLMs fail to correctly apply Bayes' rule, 
but that performance is dramatically improved by training on the predictions of a Bayesian model.  This literature
reinforces that rationality violations are present and, in some cases,
responsive to training, but also that coherence alone does not settle the
question: a model can satisfy axioms of rationality while pursuing objectives that are
misaligned.

Most of this work is diagnostic, documenting the presence or absence of
particular rationality properties.  
I argue that classic representation theorems in decision theory
provide a natural foundation for more systematic analysis.
A representation
theorem is a result of
the form: behavior satisfies axioms $A_1, \ldots, A_k$ if and only
if that behavior can be rationalized by some objective, e.g. a probability measure,
a utility function, or both.\footnote{The term ``representation'' here refers to the
rationalization of behavior by a mathematical objective, not to the feature
representations or embeddings discussed in machine learning.}  One direction says that any
``rational'' agent (in the sense specified by the theorem) satisfies the axioms, which is expected and consistent with the diagnostic literature discussed above.  
The theorems also cover the converse, however, establishing
 that if a model's behavior satisfies the
axioms, then there exists a well-defined objective that rationalizes
it.

Representation theorems are potentially powerful aids for model evaluation and training.
By construction, the axioms completely characterize the implications of a given
rationality notion for observable behavior.  Hence, if a model satisfies all the
axioms, it follows that its behavior is consistent with a given notion of rationality.
Moreover, if axiom violations can be driven to zero, this ensures that the model's choices
are consistent with the desired rationality notion without specifying what prior or utility function the model
should have.

Classical axiomatizations, such as \citet{savage1954}'s,
characterize behavior on a rich choice domain,
and in particular contemplate an agent's preferences in every conceivable
comparison.  Indeed some axioms, such as continuity conditions, have no content for
 finite collections of options.  
 For model evaluation and training, it is thus
particularly useful to consider results in the revealed preference tradition,
which characterize the implications of rationality axioms for finite datasets \citep[see][for a review]{chambers2016revealed}.
Using these results, one can check axiom compliance using only the model's own responses to finite collections of queries.  
One
could, for instance, generate synthetic choice problems (betting scenarios, budget sets, portfolio
allocations), query the model, and check whether the responses satisfy the
relevant axioms.  No external labels, ground truth outcomes, or human feedback
are required.  The combinatorial explosion that limits axiom-based testing for
human subjects may be a feature in the LLM context, where synthetic problems and responses can be
generated cheaply and in large numbers.  Moreover, as I discuss below many revealed preference
results come with continuous violation measures that can be computed in polynomial time.

It is important to emphasize that coherence is not sufficient for good behavior.  
A model satisfying all the axioms
of subjective expected utility, for instance, behaves in accordance with some 
well-defined prior and
utility function but both could be terrible, so other evaluation criteria and training signals
remain essential.

I discuss three instantiations of this idea, in order of increasing
richness.  \Cref{sec:definetti} considers probabilistic coherence through de
Finetti's theorem: a model's probability assessments are consistent with a
well-defined probability measure if and only if they cannot be Dutch-booked.
\Cref{sec:afriat} considers preference rationality through Afriat's theorem:
a model's choices from budget sets are consistent with utility maximization if
and only if they satisfy the Generalized Axiom of Revealed Preference.
\Cref{sec:seu} considers the joint structure of beliefs and preferences under
uncertainty through results of \citet{echenique2015savage}: a model's
portfolio choices are consistent with subjective expected utility maximization (with a concave utility)
if and only if they satisfy the Strong Axiom of Revealed Subjective Expected
Utility.  \Cref{sec:implementation} briefly discusses implementation.
\Cref{sec:related} situates the approach relative to RLHF, calibration, and
other suggestions from the literature.  \Cref{sec:conclusion} discusses limitations and concludes.

\section{Probabilistic Coherence: de Finetti}\label{sec:definetti}

The simplest instance of the general principle concerns probabilistic beliefs.
De Finetti's coherence theorem provides an ``if and only if''
characterization of when a set of probability assessments is consistent with a
well-defined probability measure.

\paragraph{Setup}

Consider a finite collection of events $E_1, \ldots, E_n$ defined on a state space $\Omega$.  A
\emph{prevision assignment} (i.e., a set of probabilistic predictions) is a
function $p$ that assigns a value $p(E_i) \in [0,1]$ to each event.  In our
context, these are the model's stated probability assessments when queried
about each event.

\paragraph{Avoiding sure loss}

A bet on event $E_i$ with stake $b_i \in \mathbb{R}$ yields a payoff of
$b_i(\mathbf{1}_{E_i}(\omega) - p(E_i))$ to the bettor when the state is $\omega$, 
where $\mathbf{1}_{E_i}(\omega)$ takes the value one when $\omega\in E_i$ and zero otherwise.\footnote{The convention is
that the AI agent takes the opposite side of the bet, so the agent's payoff from
this bet is $-b_i(\mathbf{1}_{E_i}(\omega) - p(E_i))$.}  A \emph{Dutch book} against
the prevision assignment $p$ is a collection of stakes $(b_1, \ldots, b_n)$
such that the agent's net payoff $-\sum_{i=1}^n
b_i(\mathbf{1}_{E_i}(\omega) - p(E_i))$ is strictly negative in every state
of the world $\omega$.  The existence of a Dutch book means that the agent's
stated probabilities are exploitable: an adversary can construct a combination
of bets that guarantees a profit regardless of which events occur.
The prevision assignment $p$ \emph{avoids sure loss} if no Dutch book exists
against it.

\paragraph{The representation theorem}

De Finetti's coherence theorem  \citep{definetti1937, definetti1974} states that the following are equivalent:
\begin{enumerate}
    \item The prevision assignment $p$ avoids sure loss.
    \item The prevision assignment $p$ can be extended to a finitely additive
    probability measure on the algebra generated by $E_1, \ldots, E_n$.
\end{enumerate}
Condition (1) depends only on the agent's stated probabilities
and can be checked without knowing which events actually occur.  Condition
(2) says that the probabilities can be rationalized by a coherent probability
measure.  The equivalence between (1) and (2) is the representation theorem.
Note that finitely additive probability suffices since we consider only a finite
collection of events.

\paragraph{The penalty}

The magnitude of the best Dutch book against a prevision assignment $p$
provides a natural, continuous penalty.  Let $\omega_1, \ldots, \omega_m$
denote representative states, one for each atom of the algebra generated by $E_1, \ldots, E_n$.  The
bettor's profit in state $\omega_j$ from a collection of stakes $(b_1,
\ldots, b_n)$ is
\[
\Pi(\omega_j) = \sum_{i=1}^n b_i \bigl(\mathbf{1}_{E_i}(\omega_j) -
p(E_i)\bigr).
\]
The sure-loss magnitude is the optimal value of the linear program
\begin{equation}\label{eq:dutchbook}
L(p) = \max_{b \in \mathbb{R}^n} \min_{j=1,\ldots,m} \;
\Pi(\omega_j),
\end{equation}
subject to $\sum_i |b_i| \leq 1$.  The normalization bounds the total size of
all bets, so $L(p)$ measures the guaranteed profit per unit of total stake.
By the representation theorem, $L(p) = 0$ if and only if $p$ is coherent.
When $L(p) > 0$, its value quantifies the exploitability of the model's stated
probabilities.  This is a linear program and is solvable in polynomial
time (in the number of atoms and events).\footnote{\citet{garrabrant2016logical} use a related no-arbitrage
criterion, requiring that no computable trading strategy earn unbounded
profits, as a coherence condition for a logical probability assigner.  Their
setting (logical statements) differs from the events considered here,
but the underlying principle is the same: the absence of profitable
exploitation characterizes coherence.}

\paragraph{Translation to the LLM context}

To apply this, one may generate a partition of a sample space into
atoms,
along with a collection of events defined as unions of those atoms (e.g. that it rains tomorrow in Boston, but not in New York,
and there is at least one hour of sun during the day in Boston, and ...).
The model is presented with descriptions of the events and asked to report a
probability for each.  The logical relationships among events (set inclusion,
partitioning, complementation) are known by construction, so the full LP can
be assembled from the model's responses without any external labels.
The resulting LP then checks all implications of the probability axioms.

An important design consideration is that the axiom checks test consistency
across a collection of assessments attributed to a single agent.  In the LLM
context, it is natural to fix a role for each batch of queries
(e.g., ``you are a forecaster assessing weather probabilities'') and to
test coherence within that role, rather than across different roles, since
it is not clear that probability assessments should be consistent across
distinct roles.

\paragraph{Existing evidence}

\citet{zhu2024incoherent} directly measure probability coherence
violations in LLM outputs across several model families, finding systematic
incoherence in probability judgments.
\citet{paleka2025consistency} construct consistency checks linking logically
related forecasts (negations, conjunctions, conditional probabilities), and
measure violations by an arbitrage metric: the worst-case performance improvement
which could be guaranteed by an adversary who aims to improve the model's forecasts.
They document substantial violations even for models
using extended reasoning, and find that consistency predicts ground-truth
forecast accuracy.  See also the applied forecasting system of
\citet{alur2025aia}, which reconciles disparate forecasts of the same event
using a supervisor agent at inference time.
\citet{betz2023neural} demonstrate that self-training improves the
probabilistic coherence of neural language models trained on synthetic data.
\citet{qiu2026bayesian} train models for Bayesian probabilistic reasoning via
supervised fine-tuning on demonstrations from an external model, and show that this 
generates performance gains which generalize to new tasks.
\citet{kim2026drift} document violations of a martingale property of LLM
predictive beliefs under resampling, and reduce them by fine-tuning on a
self-consistency loss.  Similarly, \citet{chandak2025forecasting}
find that reinforcement learning on forecasting outcomes reduces arbitrage
violations on the benchmark of \citet{paleka2025consistency} even though
consistency is not targeted directly, and \citet{lee2025advancing} propose
consistency-based auxiliary rewards for forecaster training.
These results demonstrate that probabilistic coherence can be improved by
training; the present approach differs in using a continuous penalty grounded
in a representation theorem rather than supervised examples or checks for a
subset of coherence properties.

\section{Preference Rationality: Afriat}\label{sec:afriat}

The second instantiation concerns preferences over bundles of goods.  Afriat's theorem
provides a revealed preference test for utility maximization that is
analogous to de Finetti's theorem for probabilistic coherence: behavior is
rationalizable if and only if it satisfies an axiom stated on the observed
choices.

\paragraph{Setup}

The data consist of $T$ observations $\{(x_t,
p_t)\}_{t=1}^T$, where $x_t \in \mathbb{R}^K_+$ is the bundle chosen at prices 
$p_t \in\mathbb{R}^K_{++}$ and income $w_t = p_t \cdot x_t$.

\paragraph{GARP}

The choices $\{(x_t,p_t)\}_{t=1}^T$ reveal information about
preferences.  If $x_t$ was chosen when $x_s$ was affordable
($p_t \cdot x_s \leq p_t \cdot x_t$), then $x_t$ is \emph{directly revealed
preferred} to $x_s$, written $x_t \mathrel{R^D} x_s$.  The \emph{revealed
preference} relation $R$ is the transitive closure of $R^D$ - for instance, if $x_t R^D x_s$ and $x_s R^D x_r$ then $x_t R x_r$.

The \emph{Generalized Axiom of Revealed Preference} (GARP) requires: if $x_t
R x_s$, then $p_s \cdot x_t \geq w_s$.  That is, if $x_t$ is
revealed preferred to $x_s$ (directly or through a chain), then $x_t$ must
not have been strictly inside the budget set at which $x_s$ was chosen.

\paragraph{The representation theorem}

Afriat's theorem \citep{afriat1967, varian1982} states that the following are
equivalent:
\begin{enumerate}
    \item The data $\{(x_t,p_t)\}_{t=1}^T$ satisfy GARP.
    \item There exist numbers $\{u_t, \lambda_t\}_{t=1}^T$ with $\lambda_t >
    0$ satisfying the \emph{Afriat inequalities}: \\$u_s - u_t \leq \lambda_t
    p_t \cdot (x_s - x_t)$ for all $s, t$.
    \item There exists a continuous, monotone, concave utility $U :
    \mathbb{R}^K_+ \to \mathbb{R}$ such that 
    \[
    x_t \in \argmax_{x : p_t
    \cdot x \leq w_t} ~ U(x) \text{ for all }t.\]
\end{enumerate}
The equivalence between (1) and (3) is the representation theorem: choices are
consistent with utility maximization if and only if they satisfy GARP.  If any
locally nonsatiated (i.e., for every bundle there is another nearby that is strictly preferred)
utility function rationalizes the data, a continuous, monotone, concave one
does as well \citep{afriat1967}.  The Afriat inequalities in (2) are a
finite system of linear inequalities providing the constructive check 
for whether a finite dataset is consistent with GARP.
Note that the rationalizing utility is not unique: many different utility
functions may be consistent with the same finite dataset.

\paragraph{The penalty: 1-CCEI}

GARP is a binary condition: it either holds or it fails.  For evaluation
and regularization, a continuous measure of departure is potentially helpful.  The
\emph{Critical Cost Efficiency Index} (CCEI), introduced by
\citet{afriat1973}, provides one such measure.  Define \emph{$e$-GARP}
as the relaxation in which budget constraints are tightened by a factor $e
\in [0,1]$: $x_t$ is directly revealed preferred to $x_s$ under $e$-GARP
only if $p_t \cdot x_s \leq e \cdot p_t \cdot x_t$.  The CCEI is the
largest $e$ such that the data satisfy $e$-GARP:
\[ \text{CCEI} = \sup\{e \in [0,1] : \{( x_t,p_t)\}_{t=1}^T \text{ satisfies } e\text{-GARP}\}. \]
If the data satisfy GARP then $\text{CCEI} = 1$. 
CCEI is continuous in the model's choices, and is computable
by binary search over $e$, with each step requiring a GARP check.
Unfortunately, however, when  a revealed preference comparison holds with exact
budget equality, the CCEI may equal one even though GARP fails.

\citet{echenique2021} raises other
substantive concerns about common interpretations of the CCEI.  
Alternative measures that avoid these concerns are the \emph{money pump indices} of
\citet{echenique2011mpi} and \cite{smeulders2013}.  These indices measure how much money an adversary could in principle extract using a model's
 irrational preferences, and are positive exactly when GARP fails and therefore resolve one natural shortcoming of the CCEI. 
 At the same time, however, they can change discontinuously in the agent’s chosen bundles $x_t,$ which seems potentially unappealing for our purposes.
  
  While the interpretive limitations of the CCEI raised by \citet{echenique2021} are inherent to
  the measure, the fact that the CCEI may equal one even when GARP fails can be addressed by a suitable elicitation strategy.

\begin{lemma}\label{lem:generic}
Suppose that the price vectors $p_t$ are drawn independently across
observations from distributions with densities with respect to
Lebesgue measure, that each choice $x_t$
depends only on observation $t$'s prices $p_t$ and income $w_t$, with income
fixed or drawn independently of the prices, and that choices
exhaust their budgets, $p_t \cdot x_t = w_t > 0$ for all $t$.
Then with probability one $\text{CCEI} = 1$ if and only if the
data satisfy GARP.
\end{lemma}

The proof is in the appendix.

\paragraph{Translation to the LLM context}

To use this approach for model evaluation or training, one could present the model with
a role and a budget constraint (prices and income), and ask it to allocate across goods.
One could then vary prices and income, compute the CCEI of the
resulting choices, and penalize $1 - \text{CCEI}$.
 The decision context can be varied across batches: one might
involve grocery purchases, another retirement portfolio allocation, another
time allocation across tasks.  GARP requires consistency within each context,
and testing across many contexts strengthens the regularization.  As in the de
Finetti case, it seems natural to fix a role or persona within each batch and test
consistency within that role.

\paragraph{Existing evidence}

\citet{chen2023pnas} apply the GARP framework to test whether GPT models make
rational choices.  They present GPT-3.5 with budget allocation
tasks across four domains (risk, time, social preferences, and food) and
compute the CCEI.  The average CCEI values exceed 0.997 across all domains,
outperforming human subjects in comparable experimental
designs.  Importantly, however, these tests use
two-good settings with a relatively limited number of budget sets.
Thus, while 
\citet{chen2023pnas} demonstrate that their tests have sufficient power for their purposes, GARP-based
regularization may have more bite in richer settings, where it is much easier to gather data from LLMs than from
human subjects.  \citet{wen2025specialization} finds that prompt-based role
specialization (assigning domain-specific personas such as ``biotechnology
expert'' or ``economist'') substantially reduces GARP compliance.
\citet{seror2024moral} extend GARP testing to moral preference domains,
finding variation in choices across models.  \cite{aguiar2026garp} show that fine-tuning
a time-series foundation model on synthetic data satisfying GARP substantially 
improves performance for predicting consumer choices.

\section{Decision Under Uncertainty: Echenique-Saito}\label{sec:seu}

The preceding sections address beliefs and preferences separately.  However, a model
could have coherent probabilities and rational preferences but still fail to
integrate them.  The third
instantiation of the general principle addresses this by asking if a model's choices under
uncertainty are jointly rationalized by some prior belief and utility function.

\paragraph{Setup}

Let $S = \{1, \ldots, \bar{s}\}$ be a finite set of states.  A
\emph{monetary act} is a vector $x \in \mathbb{R}^S_+$ specifying a payoff $x(s)$ in
each state $s$.  The data consist of $T$ observations $\{(x_t,
p_t)\}_{t=1}^T$, where $x_t \in \mathbb{R}^S_+$ is the portfolio chosen at
prices $p_t \in \mathbb{R}^S_{++}$ for state-contingent claims (securities
that pay one unit in a designated state and zero otherwise) with income $w_t = p_t \cdot x_t$.\footnote{The restriction
to monetary payoffs (one good per state) is a limitation relative to the
Afriat setting, which handles bundles of multiple goods.  This is, however,
a common setting for considering decision under uncertainty, particularly in asset pricing.
In addition, the extension to state-dependent utility discussed below considerably generalizes the result.}

\paragraph{SARSEU}

\citet{echenique2015savage} introduce the following condition.

\begin{definition}[SARSEU]
A dataset $\{(x_t, p_t)\}_{t=1}^T$ satisfies the \emph{Strong Axiom of
Revealed Subjective Expected Utility} (SARSEU) if, for any sequence of pairs
$(x_{t_i}(s_i), x_{t'_i}(s'_i))_{i=1}^n$ such that
\begin{enumerate}
    \item[(i)] $x_{t_i}(s_i) > x_{t'_i}(s'_i)$ for all $i$,
    \item[(ii)] each state $s$ appears as $s_i$ (on the left) the same number
    of times as it appears as $s'_i$ (on the right),
    \item[(iii)] each observation $t$ appears as $t_i$ (on the left) the same
    number of times as it appears as $t'_i$ (on the right),
\end{enumerate}
we have
\begin{equation}\label{eq: cycles}
\prod_{i=1}^n \frac{p_{t_i}(s_i)}{p_{t'_i}(s'_i)} \leq 1.
\end{equation}
\end{definition}
See \cite{echenique2015savage} for
discussion of these conditions.

\paragraph{The representation theorem}

\citet{echenique2015savage} show the following are equivalent:
\begin{enumerate}
    \item The data $\{(x_t,p_t)\}_{t=1}^T$ satisfy SARSEU.
    \item There exist a full-support probability measure
$\mu \in \Delta_{++}(S)$ and a concave, strictly increasing function $u :
\mathbb{R}_+ \to \mathbb{R}$ such that
    \[
    x_t \in \argmax_{x : p_t
    \cdot x \leq w_t} \sum_{s \in S} \mu(s) u(x(s)) \text{ for all }t.\]\end{enumerate}

Condition (2) requires that the agent has a single probability distribution over states (beliefs) and a
single utility function over monetary payoffs (preferences), and chooses $x$ to
maximize expected utility. \citet{savage1954} provided a set of axioms for preferences over ``acts'' mapping states to
``consequences'' (potentially more general than simple monetary payoffs) and
proved that, under these axioms, preferences admit an SEU representation.
  Savage's framework, however, restricts the preference ordering over all acts, which is not what one observes in
finite data.  \citet{echenique2015savage} provide the operational analog for the problem with monetary payoffs: 
the SARSEU restriction in (1) is a
revealed-preference characterization of SEU that applies to finite datasets of
choices at prices, exactly as \citet{afriat1967} operationalizes utility
theory for finite budget-set data.

As in the de Finetti and Afriat cases, the condition is necessary and
sufficient: the axiom holds if and only if behavior is consistent with subjective expected utility maximization (with a concave utility and full-support belief).
\citet{echenique2015savage} show in the proof of their Proposition 2 that SARSEU can be tested
by checking feasibility of a linear program.

SARSEU is a binary condition, so, as in \cref{sec:afriat} a continuous
measure of departure may be helpful.  The construction behind the CCEI adapts
directly: weaken the condition being checked until the data
satisfy it.  For $e \in
[0,1]$, say that the data satisfy \emph{$e$-SARSEU} if \eqref{eq: cycles}
holds for every sequence of pairs satisfying conditions (ii) and (iii)
and, in place of condition (i), $e \cdot x_{t_i}(s_i) > x_{t'_i}(s'_i)$
for all $i$.  Shrinking $e$ removes comparisons, so $e$-SARSEU only
becomes easier to satisfy, and $0$-SARSEU holds vacuously.  Let
\[ \text{E} = \sup\{e \in [0,1] : \{(x_t, p_t)\}_{t=1}^T \text{ satisfies } e\text{-SARSEU}\}, \]
and take $1 - \text{E}$ as the penalty.

\begin{lemma}\label{lem:E}
For any dataset $\{(x_t, p_t)\}_{t=1}^T$:
\begin{enumerate}
\item[(i)] The supremum defining $E$ is attained, and $E = 1$ if and only
if the data satisfy SARSEU.
\item[(ii)] $E$ equals either $1$ or one of the finitely many ratios
$x_{t'}(s')/x_{t}(s) < 1$ formed from pairs of payoffs, and is computable
by binary search over these candidate values, with each step solving the
linear program of \citet{echenique2015savage} restricted to pairs with $e \cdot x_{t}(s) >
x_{t'}(s')$.
\item[(iii)] At fixed prices, $E$ is continuous in the payoffs
$\{x_t(s)\}$ on the region where $x_t(s)>0$ for all $t$ and $s$.
\end{enumerate}
\end{lemma}

The proof is in the appendix.  Because the comparisons that trigger
SARSEU involve payoffs alone, with prices entering only through the
products in \eqref{eq: cycles}, the penalty $1 - \text{E}$ is continuous in the
model's choices at any fixed prices, on the region where
payoffs are strictly positive, while remaining zero exactly when SARSEU
holds.

An alternative measure from the literature is the minimal perturbation
index of \citet{echenique2023approximate}.  They define a dataset to be
perturbed SEU rational if it can be rationalized after the beliefs,
prices, or state-utility weights faced at different observations are
adjusted, with relative adjustments bounded by a common factor.  They
show that the three forms of perturbation are equivalent, and propose the minimal factor, written $1 +
e^*$, as a measure of the distance from SEU.  The index 
$e^*$ is positive exactly when SARSEU fails, but it depends on the
model's choices only through the ordering of the payoffs, and so is
unchanged by any perturbation of the choices that preserves that
ordering. This implies, however, that $e^*$ is piecewise constant,
which is a potential downside for present purposes.

\paragraph{Translation to the LLM context}

The elicitation procedure is similar to that of \cref{sec:afriat}.  Present the
model with a set of states of the world.  Offer state-contingent claims at given prices and budgets.  
The model allocates a budget across
states.  Multiple rounds with different price vectors yield the dataset
$\{(x_t, p_t)\}_{t=1}^T$, from which one computes E and penalizes
$1 - \text{E}$. 
 As in the preceding sections, the model should be
assigned a fixed role within each batch, since the approach tests consistency of choices
attributed to a single decision-maker.

\paragraph{Existing evidence}

Revealed preference tests of SEU rationality (SARSEU or related
conditions) on LLM-generated data appear not to exist.
\citet{mazeika2025utility} find that large models behave broadly in line with expected utility, 
but do not test the joint belief-preference structure that
SARSEU captures.  \citet{tak2026sparks} likewise test model adherence to the 
expected-utility axioms, and they further show that models exhibit ambiguity aversion in settings
with unknown probabilities.
\citet{yamin2026validating} elicit probability estimates from LLMs in decision
 tasks and test whether these beliefs satisfy properties required of a
rational decision-maker, including decision-sufficiency (i.e. that stated beliefs
fully explain the model's choices).
They find violations, suggesting that
LLM-reported probabilities do not form a fully coherent basis for the decisions
the models actually make.  This is suggestive of the kind of joint incoherence
that SARSEU is designed to detect, though their framework (random utility
models with elicited beliefs) differs from the revealed-preference approach
taken here.

\paragraph{Extensions: alternative axiom systems}

Full SEU may be too narrow a rationality standard.  In particular, SEU rules out state-dependent 
utility $u_s(\cdot)$, as could arise if the agent anticipates higher needs in some states
than others.  It also rules out ambiguity aversion, the tendency to prefer known risks over unknown ones
\citep{knight1921}, which some might view as desirable for a system facing genuine
uncertainty.

Fortunately, existing results suggest paths forward in both cases.
\citet{echenique2015savage} provide a generalization of their results to settings with state-dependent
utility, which again gives if and only if conditions for a finite dataset to be compatible with
subjective expected utility maximization with utility $u_s(\cdot)$.
Similarly, representation theorems are available for axiom systems which allow
ambiguity aversion.
\citet{gilboa1989maxmin}  obtain a maxmin expected-utility representation: the agent
maximizes expected utility under the worst-case probability measure from a
convex set of priors.  \citet{mmr2006variational} characterize variational
preferences, a broader class that includes maxmin EU as a special case. 
These, however, are again axiomatizations stated on
rich domains of choices.  The finite-data tests that the present approach
requires were developed more recently, in the GRID framework of
\citet{pqr2020}; see \citet{echenique2020are} for a survey.
\citet{dembo2026ellsberg} implement tests of this kind on human
portfolio-choice experiments, testing subjective expected utility and its
ambiguity-averse alternatives.

\section{Implementation Considerations}\label{sec:implementation}

\paragraph{The suggested procedure}

The procedure for each of the three instantiations consists of three steps:
(1) generate synthetic choice problems (collections of
events for de Finetti, budget sets for Afriat, portfolio allocation problems
for Echenique-Saito); (2) elicit responses from the model; (3) compute the
penalty (the Dutch book magnitude $L(p)$, or $1 - \text{CCEI}$, or
$1 - \text{E}$). 
The resulting penalty can then be reported directly  as
an evaluation metric, or
 used in the same way as other
constraint-based penalties.

\paragraph{Practical considerations}

The choice problems used for axiom testing need not be sampled in any specific way (save for the conditions needed to apply Lemma \ref{lem:generic}).  
An adversarial approach to problem generation seems potentially appealing, with 
a second model proposing new choice problems based on previous performance (including the
output from previous penalty computations).

An important practical concern is that the same formal choice problem,
described in different words, or even the same words in a different order, 
may elicit different responses from the model.
In a controlled training setup, this is directly testable.  Generate
the formal problem first, then produce multiple linguistic descriptions of it
and include the model's responses across multiple descriptions.
Axiom violations arising from such paraphrases ``count,'' and will be penalized by the approaches suggested above.

\paragraph{Computational costs}

The Dutch book program
is a single linear program, whose cost is polynomial in the number of atoms and events,
while the CCEI and E are computed by binary search, with each step
requiring a GARP check for the former and a linear program for the latter.  The Dutch book penalty is continuous in the model's
stated probabilities, and $1 - \text{CCEI}$ and $1 - \text{E}$ are continuous in
the model's (interior) allocations at fixed elicitation prices.

\section{Relation to Existing Approaches}\label{sec:related}

The approach proposed here complements, rather than replaces, existing
methods.

\paragraph{RLHF and preference structure.}
Current training methods including reinforcement learning
from human feedback \citep[RLHF; e.g.][]{ouyang2022rlhf} and direct preference
optimization \citep[DPO; e.g.][]{rafailov2023dpo} extract training signals from feedback.
 \citet{ge2024axioms} provide an axiomatic critique of RLHF and propose an alternative, 
axiomatically grounded approach.  All of these methods differ fundamentally from the 
 present paper in that their focus is on generating reward signals for model training, 
rather than enforcing coherence on the model outputs.  Since coherence of model outputs
is in important respects value-neutral, the approach discussed in the present paper is not a substitute
for such evaluation criteria or reward signals.

\paragraph{Calibration and proper scoring rules.}
Calibration methods \citep{guo2017calibration} check whether predicted
probabilities match empirical frequencies, requiring ground truth outcomes.
De Finetti coherence checks internal consistency without ground truth.  The
two are complementary: a model can be well-calibrated on many events but incoherent, or
coherent but poorly calibrated.  
Similarly, a strictly proper scoring rule incentivizes truthful reporting of forecast distributions \citep{gneiting2007jasa}.  
Evaluating proper scoring rules again requires external outcomes data and is
complementary to the coherence approach.

\paragraph{The diagnostic  and correction literature.}
Many authors including \citet{betz2023neural}, \citet{chen2023pnas}, \citet{hagendorff2023llm}, \citet{zhu2024incoherent}, \citet{chadwick2025dutch}, \cite{wen2025specialization}, \citet{paleka2025consistency}, and \citet{qiu2026bayesian}
have documented LLM rationality violations.  As discussed above, some of these papers also discuss how such violations may be reduced, for instance through additional training steps as in \citet{betz2023neural} and \citet{qiu2026bayesian}, or through post-processing of model outputs as in \citet{chadwick2025dutch}.
Relative to these approaches, the approach in this note is motivated by the ``if and only if''
guarantees provided by the representation theorems, which imply that if particular penalty terms could be driven to zero
then model behavior would necessarily be fully rationalizable by a prior, a utility function, or both.
The same guarantees matter for measurement: benchmarks in this
literature examine selected axioms or properties, whereas the checks used
here are exact, in that passing exhausts the implications of the maintained
rationality standard for the elicited data.

\paragraph{Consistency as a training objective.}
\citet{pres2026consistency} argue that many desiderata for language models
are naturally expressed as consistency requirements linking responses across
related inputs, and propose optimizing such cross-input consistency
functions directly, through hard constraints, soft penalties, or posterior
regularization.  The penalties considered here are instances of that
program.  The representation theorems sharpen the recommendations of a 
 generic consistency
framework by identifying which consistency conditions to impose and
guaranteeing that driving violations to zero yields behavior rationalizable by
a well-defined objective.

\paragraph{Utility engineering.}
The utility engineering agenda of \citet{mazeika2025utility} concerns
measuring and controlling the objectives implicit in LLM behavior, motivated by the finding
that larger models develop increasingly coherent value systems, but these emergent objectives may be
misaligned.
Subsequent stress tests find substantial residual incoherence 
\citep{ajayi2026incoherent}, so the strength of the
emergence finding remains contested.
This agenda appears highly
complementary to the approach I suggest: utility engineering appears most applicable in settings where 
 the model acts coherently
enough for a rationalizing objective to exist, which is what the axiom-based
approach proposed here examines.
Conditional on passing the rationality checks proposed here, one could go further
and examine the set of priors and/or utility functions consistent with observed choices,
which are natural inputs for a  utility engineering exercise.  In this spirit, \citet{yamin2026steering}, in follow-up work to
\citet{yamin2026validating}, estimate utility functions based on LLM choices
in diagnostic tasks and study how to steer models toward user-specified
objectives.

\section{Limitations and Conclusion}\label{sec:conclusion}

It is important to be clear about some limitations of the approach proposed here.

\begin{enumerate}

\item Coherence is not enough.  A model that satisfies all the axioms of
subjective expected utility has a well-defined probability measure and utility
function, but these could be arbitrary.\footnote{More broadly, some authors \citep[e.g.][]{tan2025beyond} dispute whether rational choice frameworks of the sort presumed here are an appropriate foundation for LLM behavior in the first place.}  The role of the penalties proposed
here is measurement and regularization, not a standalone objective.

\item The choice of axiom system matters.  Full subjective expected utility
rules out ambiguity aversion and other departures that may be reasonable or
desirable.  Weaker axiom systems (maxmin EU, variational preferences) impose
weaker rationality conditions and permit a broader range of behavior.

\item The joint belief-preference test of \cref{sec:seu} applies to monetary
payoffs (one good per state).  Extending to richer consequence spaces, for
example bundles of goods, is a potentially useful
direction for future work.

\end{enumerate}

Classic results from decision theory characterize, in precise and testable terms, what it
means for behavior to be rational.  These results take the form of axioms
checkable from behavior alone, with representation theorems guaranteeing that
compliance implies the existence of a rationalizing objective.  
This structure, particularly combined with results in the revealed preference tradition,
 appears well-suited to LLM
evaluation and training, where behavior can be observed at scale and tested systematically.
The connection seems potentially fruitful in both directions: decision theory provides
a natural collection of formal tools for evaluation and regularization, while
the AI model training context provides a potentially important setting for decision theorists interested in revisiting
and extending this toolkit.
 
 \appendix
\section*{Appendix: Proofs}

\begin{proof}[Proof of \cref{lem:generic}]
Fix two distinct observations $t$ and $s$.  Under the stated independence
assumptions, $p_t$ is independent of $(x_s,w_t)$: the choice $x_s$ depends only
on $(p_s,w_s)$, the price vectors are independent across observations, and $(w_1,...,w_T)$ is
 independent of the prices.  Moreover,
budget exhaustion and $w_s>0$ imply $x_s\neq 0$.  Conditional on $(x_s,w_t)$,
the event $ p_t\cdot x_s=w_t$
confines $p_t$ to the proper affine hyperplane
$\{q\in\mathbb R^K:q\cdot x_s=w_t\}$.  Because $p_t$ has a density, this event
has conditional, and hence unconditional, probability zero.  Taking a union
over the finitely many ordered pairs $(t,s)$ with $t\neq s$ shows that, with
probability one, $ p_t\cdot x_s\neq w_t$ for every $t\neq s.$

Work on this no-ties event.  If the data satisfy GARP, then they satisfy
$1$-GARP, so $\text{CCEI}=1$.  Conversely, suppose that GARP fails.  After
removing repetitions from a violating revealed-preference chain, there are
 distinct observations $t_1,\ldots,t_m$, with $m\geq 2$ and
$t_{m+1}=t_1$, such that
\[
    p_{t_i}\cdot x_{t_{i+1}}\leq p_{t_i}\cdot x_{t_i}=w_{t_i}
    \qquad (i=1,\ldots,m),
\]
and at least one of these inequalities is strict.  Since every comparison in
this cycle is between distinct observations, the no-ties property makes all of
them strict.  Define
\[
    \rho
    =\max_{i=1,\ldots,m}
      \frac{p_{t_i}\cdot x_{t_{i+1}}}{p_{t_i}\cdot x_{t_i}}.
\]
Then $\rho<1$.  For every $e\in(\rho,1]$,
\[
    p_{t_i}\cdot x_{t_{i+1}}
    < e\,p_{t_i}\cdot x_{t_i}
    \qquad (i=1,\ldots,m),
\]
so the same cycle is a cycle of strict direct $e$-revealed-preference
relations.  It therefore violates $e$-GARP.  No $e>\rho$ can enter the set whose
supremum defines the CCEI, and hence $\text{CCEI}\leq \rho<1.$
Thus, on the no-ties event, $\text{CCEI}=1$ if and only if the data satisfy
GARP.
\end{proof}

\begin{proof}[Proof of \cref{lem:E}]
Let $\mathcal C=\{(t,s):t=1,\ldots,T,\ s\in S\}$ be the set of
observation--state cells, and let $\mathcal A=\{(a,b)\in\mathcal C^2:a\neq b\}$
be the set of ordered pairs of distinct cells.  For a cell $a=(t,s)$, write
$x_a=x_t(s)$ and $p_a=p_t(s)$.  A finite sequence of comparison pairs can be
encoded by a vector $n=(n_{ab})_{(a,b)\in\mathcal A}\in\mathbb Z_+^{\mathcal A}$,
where $n_{ab}$ records the multiplicity of the ordered pair $(a,b)$.  Conditions
(ii) and (iii) of SARSEU are equivalent to $Gn=0$
for an integer matrix $G$ whose rows impose the state- and
observation-balance restrictions (see the proof of the SARSEU result in \citealt{echenique2015savage}).  If $\delta_{ab}=\log p_a-\log p_b,$
then the logarithm of the price product in \eqref{eq: cycles} is
$\delta\cdot n$.  Consequently, a sequence violates the price-product
restriction exactly when $\delta\cdot n>0$.

For $e\in[0,1]$, call $(a,b)$ \emph{licensed} at $e$ when
$e \cdot x_a>x_b$, and let $L(e)\subseteq\mathcal A$ be the set of licensed pairs.
The data violate $e$-SARSEU if and only if there is a nonzero vector
$n\in\mathbb Z_+^{\mathcal A}$ satisfying
\begin{equation}\label{eq:e-SARSEU-integer-witness}
    Gn=0,
    \qquad
    \delta\cdot n>0,
    \qquad
    \operatorname{supp}(n)\subseteq L(e).
\end{equation}
Indeed, any sequence gives such a multiplicity vector, and any such integer
vector can be expanded into a sequence by listing each pair $(a,b)$ exactly
$n_{ab}$ times.

Define the always-feasible rational polytope (i.e. a polytope defined by rational constraints)
\[
    P_0
    =\left\{z\in\mathbb R_+^{\mathcal A}:
      Gz=0,\quad \sum_{(a,b)\in\mathcal A}z_{ab}\leq 1\right\},
\]
and, for each $e$, its face
\[
    P_0(e)
    =\left\{z\in P_0:
      z_{ab}=0\ \text{for every }(a,b)\notin L(e)\right\}.
\]
Because $0\in P_0(e)$, the linear program
\begin{equation}\label{eq:e-SARSEU-LP}
    \Phi(e)=\max\{\delta\cdot z:z\in P_0(e)\}
\end{equation}
is always feasible and has an attained optimum.  I claim that the data violate
$e$-SARSEU if and only if $\Phi(e)>0$.  If $n$ satisfies
\eqref{eq:e-SARSEU-integer-witness}, then
$z=n/\sum_{ab}n_{ab}$ belongs to $P_0(e)$ and has positive objective.  Conversely,
if $\Phi(e)>0$, some vertex $v$ of the face $P_0(e)$ has
$\delta\cdot v>0$.  Since $P_0$ is a rational polytope and a vertex of a face is
a vertex of the polytope, $v$ is rational.  Multiplying $v$ by a common
denominator produces a nonzero integer vector satisfying
\eqref{eq:e-SARSEU-integer-witness}.  This proves the claim and gives the LP
check in part~(ii), with variables outside the licensed comparison set deleted.

Let $V$ be the finite vertex set of $P_0$ and define $V_+=\{v\in V:\delta\cdot v>0\}.$
For $(a,b)\in\mathcal A$, set
\[
    \rho_{ab}
    =\begin{cases}
       x_b/x_a, & x_a>0,\\
       +\infty, & x_a=0.
     \end{cases}
\]
Because all payoffs are nonnegative, $(a,b)$ is licensed at $e$ exactly when
$e>\rho_{ab}$.  For each $v\in V_+$, define $\tau_v=\max\{\rho_{ab}:v_{ab}>0\}.$
The LP characterization above implies
\[
    \text{the data violate $e$-SARSEU}
    \quad\Longleftrightarrow\quad
    e>\tau_v\ \text{for some }v\in V_+.
\]
For the reverse implication, such a vertex is itself feasible in $P_0(e)$ and
has positive objective.  For the forward implication, a positive optimum of
\eqref{eq:e-SARSEU-LP} is attained at a vertex of the face $P_0(e)$, hence at a
vertex in $V_+$ whose support is licensed.

It follows that the set of passing values is the closed interval $[0,\text{E}]$, where
\begin{equation}\label{eq:E-vertex-formula}
    \text{E}
    =\min\left\{1,\ \min_{v\in V_+}\tau_v\right\},
\end{equation}
with the inner minimum interpreted as $+\infty$ when $V_+$ is empty.  This proves
that the supremum defining $\text{E}$ is attained.  At $e=1$, the licensed pairs are
exactly those with $x_a>x_b$, so $1$-SARSEU is SARSEU.  Hence
$\text{E}=1$ if and only if the data satisfy SARSEU, proving part~(i).

If $\text{E}<1$, the finite minimum in \eqref{eq:E-vertex-formula} is attained by some
$v\in V_+$, and $\tau_v$, being a finite maximum, equals one of the ratios
$\rho_{ab}=x_b/x_a$ on the support of $v$.  That ratio is strictly below one.
Thus E is either $1$ or one of the candidate payoff ratios in the statement.
Moreover, the passing property is monotone in $e$, and for any candidate value
it is decided by the linear program \eqref{eq:e-SARSEU-LP}.  Sorting the finitely
many candidate ratios and applying binary search therefore computes E, proving
part~(ii).

Finally, fix the prices.  Then $\delta$, $P_0$, $V$, and $V_+$ do not vary with
the payoffs.  On the region where every payoff is strictly positive, each ratio
$\rho_{ab}=x_b/x_a$ is continuous.  Each $\tau_v$ is a maximum of finitely many
such continuous functions, and \eqref{eq:E-vertex-formula} expresses E as a
minimum of finitely many continuous functions and the constant function $1$.
Therefore E is continuous on that region, proving part~(iii).
\end{proof}

\bibliographystyle{apalike}
\bibliography{references}

\end{document}